%% file: main.tex
\documentclass[11pt]{article}

\title{Connectivity Oracles Under Vertex Failures \\
via a Simple and Fast Low-Degree Steiner Forest Decomposition}

\author{Sayan Bhattacharya\thanks{University of Warwick, UK.} \and Ermiya Farokhnejad\thanks{University of Warwick, UK.} \and Thatchaphol Saranurak\thanks{University of Michigan, USA.
        } \and Haoze Wang\thanks{Peking University, China.}}

\date{}

\input{commands}

\begin{document}

\maketitle
\pagenumbering{gobble}
\input{abstract}
\newpage
\tableofcontents
\addtocontents{toc}{\protect\setcounter{tocdepth}{2}}
\newpage
\pagenumbering{arabic}

\input{introduction}
\input{Steiner-Forest-Decom}

\input{hierarchy}

\input{connctivity-oracle}

\section*{Acknowledgments}
We thank Yaowei Long for a very helpful discussion about connectivity oracles under vertex failures.

Sayan Bhattacharya is funded by the European Union (ERC grant, DYNALP, 101170133). Views and opinions expressed are however those of the author(s) only and do not necessarily reflect those of the European Union or the European Research Council Executive Agency. Neither the European Union nor the granting authority can be held responsible for them. 

Thatchaphol Saranurak is supported by NSF Grant CCF-2238138 and a Sloan Fellowship.

\bibliographystyle{alpha}
\bibliography{references}

\appendix
\input{appendix-lower}

\input{appendix-pseudocodes}

\end{document}

%% file: commands.tex
\usepackage[linesnumbered,ruled,vlined]{algorithm2e}
\usepackage{multirow}
\usepackage{fancyhdr}
\usepackage{blindtext}
\usepackage{amsthm}
\usepackage{amsmath}
\usepackage{amssymb}
\usepackage{thmtools}
\usepackage{thm-restate}
\usepackage{anyfontsize}
\usepackage{xstring}
\usepackage{xargs}
\usepackage{makeidx}
\usepackage{cancel}
\usepackage{chngcntr}
\usepackage{caption}
\usepackage{subcaption}
\usepackage{float}
\usepackage{pgf,tikz,pgfplots}
\usepackage{latexsym}
\usepackage{verbatim}
\usepackage{bookmark}
\usepackage[toc]{appendix}
\usepackage[letterpaper, margin=1in]{geometry}

\usepackage{cleveref}

\usepackage{listings}
\usepackage{hyperref}

\hypersetup{
	colorlinks=true,
	linkcolor=blue,
	filecolor=magenta,      
	urlcolor=cyan,
}

\usepackage{algorithmicx}
\usepackage{algpseudocode}
\usepackage[dvipsnames]{xcolor}

\usepackage{bidicontour}

\newtheorem{theorem}{Theorem}[section]
\newtheorem{lemma}[theorem]{Lemma}

\newtheorem{corollary}[theorem]{Corollary}

\newtheorem{definition}[theorem]{Definition}

\newtheorem{observation}[theorem]{Observation}
\newtheorem{claim}[theorem]{Claim}

\hypersetup{ colorlinks, citecolor=blue}

\PassOptionsToPackage{unicode}{hyperref}

\usepackage{framed}
\usepackage[framemethod=tikz]{mdframed}
\usepackage{tikz-cd}

\newcommand{\calT}[0]{\mathcal{T}}

\newcommand{\calC}[0]{\mathcal{C}}

\newcommand{\red}[1]{{\color{red}#1}}

\newcommand{\new}[0]{\text{new}}

\newcommand{\final}[0]{\text{final}}

\newcommand{\CTP}[0]{\textsc{CutTreePair}}

\usepackage{ifthen}
\newboolean{debug}
\setboolean{debug}{false}

%% file: abstract.tex
\begin{abstract}
    
We study the \emph{low-degree Steiner forest decomposition}. Given a graph $G=(V,E)$ and a terminal set $U\subseteq V$, the standard decomposition returns a set $X\subseteq V$ of size at most $|U|/2$ and a forest $T\subseteq G-X$ of maximum degree $\Delta$ such that, for every connected component $C$ of $G-X$, some connected component of $T$ contains all terminals in $U\cap V(C)$. This is the central decomposition behind several connectivity oracles under vertex failures \cite{DP20,LS22,long2024better}. The state-of-the-art algorithms either take $O(mn\log n)$ time with degree bound $4$ \cite{DP20}, or take $m^{1+o(1)}$ time with the weaker degree bound $O(\log^{2}n)$ \cite{long2024better}.

We show that if $T$ is allowed to contain vertices of $X$, then a degree-$4$ decomposition can be computed by a very simple algorithm in $O(m\alpha(n))$ time. Further, we show that this relaxed decomposition is equally useful for constructing connectivity oracles under vertex failures.
As a consequence, we obtain a deterministic connectivity oracle under $d$ vertex failures with $\tilde{O}(m)$ space, $\tilde{O}(md_\star)$ preprocessing time ($d_\star$ is an upper bound on the number of failed vertices), $\tilde{O}(d^{2})$ update time, and $O(d)$ query time. Up to polylogarithmic factors, this oracle strictly improves all known oracles; in particular, it removes the $n^{o(1)}$ factors from the preprocessing and update times of \cite{LS22,long2024better}.
\end{abstract}

%% file: introduction.tex
\section{Introduction}

Low-degree trees are robust certificates of connectivity under vertex failures. If a tree has maximum degree $\Delta$, deleting $d$ vertices creates at most $1+\Delta d$ connected components. This simple observation underlies most connectivity oracles under vertex failures. However, for a given graph $G=(V,E)$, a low-degree Steiner tree spanning a given terminal set $U\subseteq V$ might not always exist.

\paragraph{Low-Degree Steiner Forest Decomposition.}
This motivates a \emph{low-degree Steiner forest decomposition}. It removes a set $X\subseteq V$ with $|X|\le |U|/2$ and constructs a forest $T\subseteq G-X$ of maximum degree at most $\Delta$ such that, for every connected component $C$ of $G-X$, one connected component of $T$ contains all terminals in $U\cap V(C)$. This decomposition is the core structure in several connectivity oracles under vertex failures \cite{DP20,LS22,long2024better}. Roughly, these oracles exploit the low-degree trees which are robust against vertex failures and handle $X$ recursively with logarithmic depth since $|X| \leq |U|/2$.

\paragraph{History.}

Duan and Pettie \cite{DP20} first showed how to compute this decomposition in $O(mn\log n)$ time with degree bound $\Delta=4$, which is optimal up to a constant. Their algorithm uses F{\"u}rer and Raghavachari's approximation algorithm for the minimum-degree Steiner tree \cite{furer1994approximating} as a subroutine.\footnote{The recent faster algorithms for computing minimum-degree spanning trees \cite{bhattacharya2026additive,bhattacharya2026minimum} likely extend to minimum-degree Steiner trees. Interestingly, unlike \cite{furer1994approximating}, they do not seem directly useful for computing low-degree Steiner forest decompositions.}

Later, Long and Saranurak \cite{LS22} reduced the low-degree Steiner forest decomposition to a vertex-expander decomposition. Using the cut-matching-game framework, they obtained a faster, $m^{1+o(1)}$-time algorithm, but with the weaker degree bound $n^{o(1)}$. More recently, Long and Wang \cite{long2024better} showed that a relaxed variant of the cut-matching game, which does not produce a vertex-expander decomposition, still yields a low-degree Steiner forest decomposition with the degree bound $O(\log^{2}n)$ while retaining $m^{1+o(1)}$ running time.

This development leads to a natural question:
\begin{center}
\emph{Can we compute a low-degree Steiner forest decomposition with degree $O(1)$ in near-linear time?}
\par\end{center}

An affirmative answer would immediately improve connectivity oracles under vertex failures.

\subsection{Our Result}

We show that, by allowing $T$ to contain vertices in $X$, a very fast construction is indeed possible and is strong enough to yield a new connectivity oracle under vertex failures that subsumes all known results up to polylogarithmic factors. 

\begin{theorem}[Weak Low-Degree Steiner Forest Decomposition]\label{thm:main-decomp}
There is a simple, deterministic algorithm that, given a graph $G=(V,E)$ and a terminal set $U\subseteq V$, computes $(X,T)= \CTP(G, U)$ in $O((m+n)\alpha(n))$ time\footnote{The inverse-Ackermann factor in the running time comes from the standard disjoint-set union data structure. The running time can be replaced by $O(m+n\log n)$ using elementary implementations with lists (see Theorem 21.1 of \cite{CLRS}).} where
\begin{enumerate}
\item\label{prop:decomp-size} a set $X\subseteq V$ with $|X|\leq |U|/2$; and
\item\label{prop:decomp-forest} a forest $T\subseteq G$ with maximum degree $4$ such that, for every connected component $C$ of $G-X$, some tree $T_C$ of $T$ spans $U\cap V(C)$.
\end{enumerate}
\end{theorem}

We call the resulting decomposition \emph{weak}: the components are defined in $G-X$, but their certifying trees may leave those components and contain vertices in $X$, and the same component of $T$ may therefore serve several components of $G-X$.\footnote{This terminology parallels other weak graph decompositions. In a weak low-diameter decomposition, the short paths certifying a cluster's diameter may leave the cluster. In a weak expander decomposition, the embedding certifying a component's expansion may leave the component; see, for example, \cite{li2024congestion,fleischmann2025faster}.} 

Compared with previous decompositions, our construction is much faster, taking almost truly linear time, while guaranteeing an optimal
degree bound $O(1)$. See \Cref{tab:decomp}. 
Importantly, our algorithm is very simple (see the description in \Cref{sec:basic:algo}). In contrast, previous approaches use either a cut-matching game \cite{LS22,long2024better} or F{\"u}rer and Raghavachari's algorithm \cite{DP20} as a subroutine.

\begin{table}
\begin{centering}
\begin{tabular}{|c|c|c|c|}
\hline 
 & Degree & Time & Weak?\tabularnewline
\hline 
\hline 
\cite{DP20} & $4$ & $O(mn\log n)$ & No\tabularnewline
\hline 
\cite{LS22} & $n^{o(1)}$ & $m^{1+o(1)}$ & No\tabularnewline
\hline 
\cite{long2024better} & $O(\log^{2}n)$ & $m^{1+o(1)}$ & No\tabularnewline
\hline 
This paper & $4$ & $O((m+n)\alpha(n))$ & Yes\tabularnewline
\hline 
\end{tabular}
\par\end{centering}
\caption{Known low-degree Steiner forest decompositions.\label{tab:decomp}}
\end{table}

\subsection{A New State-of-the-Art Connectivity Oracle under Vertex Failures}

Next, we describe the application of our decomposition.
A \emph{$d_{\star}$-vertex-failure connectivity oracle} is a data structure with three phases. First, it preprocesses a graph $G$. Second, given a failure set $D\subseteq V$ of size $d\le d_{\star}$, it updates its state. Third, given two vertices $s$ and $t$, it reports whether they are connected in $G-D$. 

Among an extensive line of work \cite{duan2010connectivity,van2019sensitive,DP20,LS22,pilipczuk2021algorithms,kosinas2023connectivity,long2024better,li2026}, Long and Saranurak \cite{LS22} gave an oracle with $O(m\log^{*}n)$ space, $m^{1+o(1)}+\tilde{O}(d_{\star}m)$ preprocessing time, $d^{2}n^{o(1)}$ update time, and $O(d)$ query time. Under well-known conjectures, these four parameters are simultaneously optimal up to subpolynomial factors. A natural question is whether the $n^{o(1)}$ factors in the preprocessing and update times can be removed. Long and Wang \cite{long2024better} improved the update time to $\tilde{O}(d^{2})$, but their preprocessing time remains $m^{1+o(1)}+\tilde{O}(d_{\star}m)$.

Our decomposition removes these subpolynomial factors and yields an oracle that subsumes all known results up to polylogarithmic factors. See \Cref{tab:oracle} for a comparison.
\begin{theorem}\label{thm:main-oracle}
Given an undirected graph $G$ and an upper bound $d_{\star}$ on the size $d$ of the failure set, there exists a deterministic $d_{\star}$-vertex-failure connectivity oracle with $O(m\log^{3}n)$ space, $O(d_{\star}m\log^{3}n)$ preprocessing time, $O(d^{2}\log^{3}n\log^{4}(d\log n))$ update time, and $O(d)$ query time.
\end{theorem}

The key insight behind \Cref{thm:main-oracle} is the observation that the requirement in the low-degree Steiner forest decomposition that the forest $T$ avoids $X$ is completely unnecessary for the connectivity-oracle application. Thus, our weak decomposition from  \Cref{thm:main-decomp} immediately speeds up the preprocessing time of the known framework of \cite{LS22}. 

In more details, we first show that a weak low-degree Steiner forest decomposition produces a \emph{low-degree hierarchy}, introduced in \cite{DP20} and formally defined in \cite{LS22}, in exactly the same way as a standard low-degree Steiner forest decomposition, with no loss in the hierarchy's parameters. Plugging our faster hierarchy construction into the framework of \cite{LS22} gives \Cref{thm:main-oracle}. For completeness, we also present a simpler but less efficient oracle in \Cref{sec:combine-DP}.

\begin{table}
\centering
\resizebox{\textwidth}{!}{%
\begin{tabular}{|c|c|c|c|c|c|}
\hline 
 & Det./Rand. & Space & Preprocessing & Update & Query\tabularnewline
\hline 
\hline 
\cite{duan2010connectivity} & Det. & $\tilde{O}(md_{\star}^{1-\frac{2}{c}}\red{n^{\frac{1}{c}-\frac{1}{c\log(2d_{\star})}}})$ & $\tilde{O}(md_{\star}^{1-\frac{2}{c}}\red{n^{\frac{1}{c}-\frac{1}{c\log(2d_{\star})}}})$ & $\tilde{O}(d^{\red{2c+4}})$ & $O(d)$\tabularnewline
\hline 
\cite{van2019sensitive} & $\red{\text{Rand.}}$ & $O(\red{n^{2}})$ & $O(\red{n^{\omega}})$ & $O(d^{\red{\omega}})$ & $O(d^{\red 2})$\tabularnewline
\hline 
\multirow{2}{*}{\cite{DP20}} & $\red{\text{Rand.}}$ & $O(m\log^{6}n)$ & $O(m\red n\log n)$ & $O(d^{2}\log^{3}n\log\log n)$ w.h.p. & $O(d)$\tabularnewline
\cline{2-6} 
 & Det. & $O(m\red{d_{\star}}\log n)$ & $O(m\red n\log n)$ & $O(d^{\red 3}\log^{3}n)$ & $O(d)$\tabularnewline
\hline 
\multirow{3}{*}{\cite{LS22}} & Det. & $O(m\log^{3}n)$ & $O(m\red n\log n)$ & $O(d^{2}\log^{3}n\log^{4}(d\log n))$ & $O(d)$\tabularnewline
\cline{2-6} 
 & Det. & $O(m\log^{*}n)$ & $m^{1+\red{o(1)}}+\tilde{O}(d_{\star}m)$ & $d^{2}\red{n^{o(1)}}$ & $O(d)$\tabularnewline
\cline{2-6} 
 &  & $\Omega(m)$ & $\Omega(d^{1-o(1)}m)$ & $\Omega(d^{2-o(1)})$ & $\Omega(d)$\tabularnewline
\hline 
\multirow{2}{*}{\cite{pilipczuk2021algorithms}} & Det. & $n\red{2^{2^{O(d_{\star})}}}$ & $m\red{n^{2}2^{2^{O(d_{\star})}}}$ & $\red{2^{2^{O(d_{\star})}}}$ & $\red{2^{2^{O(d_{\star})}}}$\tabularnewline
\cline{2-6} 
 & Det. & $n^{\red 2}\red{\text{poly}(d_{\star})}$ & $\red{\text{poly}(n)2^{O(d\log d)}}$ & $\red{\text{poly}(d_{\star})}$ & $\red{\text{poly}(d_{\star})}$\tabularnewline
\hline 
\cite{kosinas2023connectivity} & Det. & $O(m\red{d_{\star}}\log n)$ & $O(md_{\star}\log n)$ & $O(d^{\red 4}\log n)$ & $O(d)$\tabularnewline
\hline 
\cite{long2024better} & Det. & $O(m\log^{3}n)$ & $m^{1+\red{o(1)}}+O(d_{\star}m\log^{3}n)$ & $O(d^{2}(\log^{7}n+\log^{5}n\log^{4}d))$ & $O(d)$\tabularnewline
\hline 
\multirow{2}{*}{\cite{li2026}} & Det. & $O(n\red{d_{\star}^{2}}\alpha_{c}(n))$ & $\red{d_{\star}^{O(d_{\star}^{2})}}n+\tilde{O}(d_{\star}m+d_{\star}^{\red 6}n)$ & $O(d^{\red 4})$ & $O(d^{\red 3})$\tabularnewline
\cline{2-6} 
 & Det. & $\red{2^{O(d_{\star}^{2})}}n+O(n\red{d_{\star}^{\red 2}}\alpha_{c}(n))$ & $\red{d_{\star}^{O(d_{\star}^{2})}}n+\tilde{O}(d_{\star}m+d_{\star}^{\red 6}n)$ & $O(d^{\red 6})$ & $O(d)$\tabularnewline
\hline 
This paper & Det. & $O(m\log^{3}n)$ & $O(md_{\star}\log^{3}n)$ & $O(d^{2}\log^{3}n\log^{4}(d\log n))$ & $O(d)$\tabularnewline
\hline 
\end{tabular}
}
\caption{Known connectivity oracles that can handle an arbitrary number of vertex failures. The red text highlights suboptimal factors.\label{tab:oracle}}
\end{table}

\paragraph{Notation.}
Throughout the paper, $G=(V,E)$ is a simple undirected graph with $n:=|V|$ vertices and $m:=|E|$ edges. For a subgraph $H\subseteq G$, let $V(H)$ and $E(H)$ denote its vertex and edge sets, respectively, and let $\deg_H(u)$ denote the degree of $u\in V(H)$ in $H$. For a collection $\calC$ of subgraphs of $G$, let $V(\calC)$ denote the set of vertices belonging to at least one member of $\calC$. For $X\subseteq V(H)$, let $H[X]$ denote the subgraph of $H$ induced by $X$. We abbreviate $G[V\setminus X]$ as $G-X$. Finally, $\alpha(\cdot)$ denotes the inverse Ackermann function.

%% file: Steiner-Forest-Decom.tex
\section{Simple Weak Low-Degree Steiner Forest Decomposition}

The goal of this section is to prove \Cref{thm:main-decomp}. 

For convenience, for any subgraph $H$ of $G$
let us call a forest $T$ a \emph{Steiner forest for $U$ with respect to $H$} if,  for every connected component $C$ of $H$ with $U \cap V(C) \neq \emptyset$, some tree $T_C$ of $T$ spans all terminals in $U \cap V(C)$.
Thus, the forest $T$ from Property \ref{prop:decomp-forest} of \Cref{thm:main-decomp} is a Steiner forest for $U$ with respect to $G-X$.
Also, we say that $T$ is a \emph{degree-$4$ forest for $U$} if $U \subseteq V(T)$ and $T$ has maximum degree at most $4$.
A vertex $u \in V(T)$ is \emph{low-degree} in $T$ if $\deg_T(u) \leq 3$.

\subsection{Algorithm Description}
\label{sec:basic:algo}

The algorithm works as follows;
\begin{itemize}
\item Initialize $T$ as $|U|$ isolated vertices, one for each terminal in $U$, i.e.,  $V(T) \gets U$ and $E(T) \gets \emptyset$.
\item Repeatedly connect two components of $T$ by a path $P=(u,\dots,v)$ whose endpoints $u$ and $v$ are low-degree vertices in $T$ and whose internal vertices lie outside $V(T)$.
We refer to $P$ as a \emph{linking path}.
Update $V(T) \gets V(T) \cup V(P)$ and $E(T) \gets E(T) \cup E(P)$.
\item Once no linking path exists, return $(X,T)$ where $X$ contains precisely the vertices of degree $4$ in $T$. 
\end{itemize}

\subsection{Correctness}

\textbf{Property \ref{prop:decomp-size}.}
Every leaf of $T$ belongs to $U$ throughout the algorithm.
In any degree-$4$ forest, the number of degree-$4$ vertices is at most half the number of leaves.
Thus, the definition of $X$ gives $|X| \leq |U|/2$.

\medskip
\noindent
\textbf{Property \ref{prop:decomp-forest}.}
By construction, $T$ is a degree-$4$ forest for $U$.
It remains to show that $T$ is a Steiner forest for $U$ with respect to $G-X$.
Let $C$ be a connected component of $G-X$ with $U \cap V(C) \neq \emptyset$.
Since $U \subseteq V(T)$, at least one component of $T$ contains a terminal in $U \cap V(C)$.
The uniqueness of this component follows from the claim below.

\begin{claim}
    For every connected component $C$ of $G-X$, at most one connected component $T_1$ of $T$ satisfies $V(C) \cap V(T_1) \neq \emptyset$.
\end{claim}

\begin{proof}
Suppose for a contradiction that two distinct components of $T$ intersect $C$.
Among all paths in $C$ joining two distinct components of $T$, let $P^\star$ be one of minimum length, with endpoints $u \in V(T_1)$ and $v \in V(T_2)$.
The vertices $u$ and $v$ are low-degree in $T$ because they do not belong to $X$.
Moreover, every internal vertex of $P^\star$ lies outside $V(T)$; otherwise, a proper subpath of $P^\star$ would join two distinct components of $T$.
Thus, $P^\star$ is a linking path, contradicting the termination of the algorithm.
\end{proof}

\subsection{Implementation Details}\label{sec:implementation}

To implement our greedy algorithm from \Cref{sec:basic:algo}, we use the following data structures:
\begin{itemize}
\item We maintain a disjoint-set union data structure to keep track of the connected components of $T$.

\item We maintain a set $V_a \subseteq V \setminus V(T)$ of {\bf active} vertices. Initially, $V(T) = U$, $E(T) = \emptyset$, and $V_a = V \setminus V(T)$. Subsequently, vertices are only deleted from $V_a$, so this set shrinks monotonically.
We store one bit per vertex to indicate whether it is active.

\item We maintain the subgraph of $G$ induced by the active vertices. We call it the {\bf active subgraph} and denote it by $G_a := (V_a, E_a)$, where $E_a = \{ (u, v) \in E : u, v \in V_a\}$. We support adjacency-list queries on $G_a$.

\item An edge is a {\bf boundary edge} if either (i) its endpoints are low-degree vertices in distinct components of $T$, or (ii) one endpoint is a low-degree vertex of $T$ and the other is active.
We maintain candidate boundary edges in a doubly linked list $E_b \subseteq E$. We allow stale entries and discard them when they are processed.
\end{itemize}

We implement the algorithm from \Cref{sec:basic:algo} in {\bf rounds}. Each round processes one boundary edge and either finds a linking path or makes progress by deleting vertices from $V_a$. The algorithm terminates when $E_b = \emptyset$.

\medskip
\noindent {\bf Implementing a round.} Choose an arbitrary edge $(u, x) \in E_b$ and consider two cases.

\medskip
\noindent {\em Case 1: Both $u$ and $x$ are low-degree vertices in distinct components of $T$.} In this case, $P = (u, x)$ is a linking path.
We set $E(T) \leftarrow E(T) \cup \{(u, x)\}$ and update $E_b$ accordingly.

\medskip
\noindent {\em Case 2: $u$ is a low-degree vertex in $T$ and $x \in V_a$.} Perform a DFS from $x$ in $G_a$. When the DFS first visits a vertex $y$, check whether $y$ has a neighbor $v$ in the input graph $G$ such that $v$ is a low-degree vertex of $T$ in a component different from the component of $u$. If so, mark $y$ {\bf relevant}; otherwise, mark it {\bf irrelevant}. Terminate the DFS upon encountering the first relevant vertex.

The DFS terminates in one of two ways:
\begin{itemize}
    \item If it explores the entire component of $x$ in $G_a$ without encountering a relevant vertex, delete every visited vertex from $V_a$ and update $G_a$ and $E_b$ accordingly.

    \item If it reaches a relevant vertex $y$ with a qualifying neighbor $v$, let $P$ consist of the edge $(u,x)$, the path in the DFS tree from $x$ to $y$, and the edge $(y,v)$.
    This is a linking path.
    Add all vertices and edges of $P$ to $T$, delete every vertex visited by the DFS from $V_a$, and update $G_a$ and $E_b$ accordingly. Then terminate the current round.
\end{itemize}
The implementation is summarized in the pseudocode in \Cref{appendix:pseudocode}.

\subsubsection{Running Time of the Implementation}

\textbf{Time Spent on Maintaining Data Structures.}
Since vertices are only deleted from $V_a$, the total time spent maintaining $G_a$ is $O(m)$.
The forest $T$ is incremental, so the disjoint-set union data structure takes $O(n\alpha(n))$ time in total.
Using appropriate pointers and disjoint-set queries, we can determine in $O(\alpha(n))$ time whether an edge $e \in E$ is currently a boundary edge, and can insert a candidate into or delete it from $E_b$ in $O(1)$ time whenever an endpoint changes status.
Vertices are only inserted into $V(T)$, their degrees in $T$ only increase, components of $T$ only merge, and vertices are only deleted from $V_a$. Hence, each edge enters $E_b$ only constantly many times, and the total time spent maintaining $E_b$ is $O(m)$.

\medskip
\noindent
\textbf{Time Spent on Scanning an Edge in $E_b$.}
Given an edge $(u,x) \in E_b$, we can determine in $O(\alpha(n))$ time whether it is stale, whether $P=(u,x)$ is a linking path, or whether a DFS should start from $x$.
Since only $O(m)$ candidates are processed, their total scanning time is $O(m\alpha(n))$.

\medskip
\noindent
\textbf{Time Spent on DFS.}
Each active vertex is scanned at most once overall, because it is removed from $G_a$ at the end of the DFS that visits it.
When visiting a vertex $z$, we can determine in $O(\deg_G(z)\alpha(n))$ time whether it is relevant.
Thus, the total time spent on all DFS calls is $O(m\alpha(n))$.
Combining these bounds, the total running time is $O((m+n)\alpha(n))$, as claimed in \Cref{thm:main-decomp}.

\subsubsection{Correctness of the Implementation}

It remains to prove that no linking path $P=(u,\ldots,v)$ remains when \Cref{alg:indirect} terminates.

\medskip
\noindent
\textbf{Intuition.}
Suppose that a linking path $P=(u,\ldots,v)$ remains at termination.
In the nontrivial case $|P|>2$, consider the first time an internal vertex $x$ of $P$ is visited by a DFS.
At that time, the DFS explores at least one of the subpaths $(x,\ldots,v)$ and $(u,\ldots,x)$ and discovers a linking path containing $x$.
The vertex $x$ would then be added to the forest, contradicting the fact that it is an internal vertex of a linking path at termination.

\medskip
\noindent
We now formalize this intuition.

\medskip
\noindent
\textbf{Structure of the DFS Tree After Termination.}
Consider a DFS started from $x_1$ while processing the edge $(u,x_1) \in E_b$.
Let $T$ be the forest when this DFS starts, let $T_u$ be the component of $T$ containing $u$, and let $T^\new$ be the forest when the DFS terminates.
Either $T^\new=T$, or $T^\new$ is obtained by adding a linking path to $T$, thereby inserting vertices $x_1,\ldots,x_q$ and edges $(u,x_1),(x_1,x_2),\ldots,(x_q,v)$.
Let $T^\new_u$ be the component of $T^\new$ containing $T_u$.
Thus, either $T^\new_u=T_u$, or $T^\new_u$ is obtained by merging $T_u$ with another component of $T$ and the newly inserted path.
Let $W$ be the set of vertices visited by this DFS, and let $W' = W \setminus V(T^\new)$.
Thus, $W'=W$ if the DFS fails, whereas $W'=W\setminus\{x_1,\ldots,x_q\}$ if it finds the linking path above.
A visited vertex $y$ belongs to $W'$ precisely when DFS$(y,T_u)$ fails to find a linking path and returns to the parent of $y$ in the DFS tree.

\begin{claim}\label{claim:no-crossing}
    No edge of $G[V \setminus V(T^\new)]$ crosses $W'$.
\end{claim}

\begin{proof}
We prove the claim by induction over the sequence of DFS executions.
Consider an arbitrary edge $(y,z) \in E(G)$ crossing $W'$ such that $y \in W'$ and $z \notin W'$.
For the sake of contradiction, assume that $z \notin V(T^\new)$.
The vertex $z$ must be inactive. Otherwise, it would have been added to the DFS tree while $y$ was visited, since DFS$(y,T_u)$ cannot fail before visiting all active neighbors of $y$.
Thus, an earlier DFS must have removed $z$ from the active subgraph.
At the end of that earlier DFS, $y$ was still active and the edge $(z,y)$ existed, contradicting the induction hypothesis.
This contradiction completes the proof.
\end{proof}

\begin{claim}\label{claim:cross-structure}
    Every edge crossing $W'$ in $G$ either has an endpoint in $V(T^\new_u)$ or has an endpoint of degree $4$ in $T^\new$.
\end{claim}

\begin{proof}
    Consider an edge $(y,z) \in E(G)$ crossing $W'$ such that $y \in W'$ and $z \notin W'$.
    By \Cref{claim:no-crossing}, $z \in V(T^\new)$.
    If $z$ entered the forest at the end of this DFS, then $z \in V(T^\new_u)$.
    Otherwise, $z$ was already in $V(T)$ when the DFS started.
    If $z$ is low-degree and belongs to a component different from $T_u$, then $y$ was relevant when visited and the DFS must have succeeded through $y$.
    This contradicts $y \in W'$, because $y$ would have entered the forest on the resulting linking path.
    Hence, either $z \in V(T^\new_u)$ or $z$ has degree $4$ in $T^\new$.
\end{proof}

\input{fig-DFS}

\medskip
\noindent
\textbf{No Linking Path Remains at the End of \Cref{alg:indirect}.}
Let $T^\final$ be the forest when \Cref{alg:indirect} terminates, and suppose for contradiction that a linking path remains.
A linking path with no internal vertex would itself be a boundary edge, contradicting termination. We may therefore write the path as $P=(u,x_1,\ldots,x_q,v)$ for some $q \geq 1$.
Let $T_u^\final$ and $T_v^\final$ be the distinct components of $T^\final$ containing $u$ and $v$, respectively.
Consider the DFS that made $x_1$ inactive (this must exist, as otherwise, the algorithm would have started another DFS from $x_1$ by scanning $(u,x_1) \in E_b$), and assume that this DFS was started while scanning $(u',x_1') \in E_b$.
For this DFS, consider the notations $T^\new$, $T_{u'}^\new$, and $W'$ as defined above.
At least one endpoint $w \in \{u,v\}$ lies outside $V(T_{u'}^\new)$, as otherwise, $u$ and $v$ would belong to the same component of $T^\final$.

Every internal node of $P$ lies outside $T^\final$, and therefore outside $T^\new$.
In particular, $x_1$ is an internal node of $P$, and so $x_1 \notin V(T^\new)$.
This concludes that $x_1 \in W'$ since $x_1$ becomes inactive in this DFS.
Consider the subpath of $P$ from $x_1$ toward $w$.
Starting from $x_1 \in W'$ and applying \Cref{claim:no-crossing} along the subpath shows that every node on the subpath before $w$ belongs to $W'$.
Finally, we have two cases;
\begin{itemize}
\item 
If $w \notin V(T^\new)$, one more application of \Cref{claim:no-crossing} gives $w \in W'$, which means that $w$ became inactive without entering the forest and can never enter it later, contradicting $w \in V(T^\final)$.
\item 
If $w \in V(T^\new)$, the last edge of the subpath crosses $W'$, and using \Cref{claim:cross-structure} either $w \in V(T^\new_{u'})$ or $\deg_{T^\new}(w) = 4$.
The first one contradicts the choice of $w$, and the second one contradicts the fact that $w$ is a low-degree node in $T^\final$ since the forest is incremental.
\end{itemize}
We have therefore proved the following corollary establishing the correctness of the implementation.

\begin{corollary}
    There is no linking path when \Cref{alg:indirect} terminates.
\end{corollary}

%% file: fig-DFS.tex
\begin{figure}[ht]
\caption{
An illustration of the DFS tree after the search terminates.}
\centering

\tikzset{every picture/.style={line width=0.75pt}} 

\begin{tikzpicture}[x=0.75pt,y=0.75pt,yscale=-0.7,xscale=0.7]

\draw [color={rgb, 255:red, 0; green, 0; blue, 0 }  ,draw opacity=1 ][line width=0.75]    (405,104.6) -- (425,147.6) ;
\draw [color={rgb, 255:red, 0; green, 0; blue, 0 }  ,draw opacity=1 ][line width=0.75]    (405,103.6) -- (448,67.6) ;
\draw  [fill={rgb, 255:red, 0; green, 0; blue, 0 }  ,fill opacity=1 ] (444,67.6) .. controls (444,65.39) and (445.79,63.6) .. (448,63.6) .. controls (450.21,63.6) and (452,65.39) .. (452,67.6) .. controls (452,69.81) and (450.21,71.6) .. (448,71.6) .. controls (445.79,71.6) and (444,69.81) .. (444,67.6) -- cycle ;
\draw  [fill={rgb, 255:red, 0; green, 0; blue, 0 }  ,fill opacity=1 ] (421,147.6) .. controls (421,145.39) and (422.79,143.6) .. (425,143.6) .. controls (427.21,143.6) and (429,145.39) .. (429,147.6) .. controls (429,149.81) and (427.21,151.6) .. (425,151.6) .. controls (422.79,151.6) and (421,149.81) .. (421,147.6) -- cycle ;
\draw  [fill={rgb, 255:red, 0; green, 0; blue, 0 }  ,fill opacity=1 ] (401,104.6) .. controls (401,102.39) and (402.79,100.6) .. (405,100.6) .. controls (407.21,100.6) and (409,102.39) .. (409,104.6) .. controls (409,106.81) and (407.21,108.6) .. (405,108.6) .. controls (402.79,108.6) and (401,106.81) .. (401,104.6) -- cycle ;
\draw [color={rgb, 255:red, 0; green, 0; blue, 0 }  ,draw opacity=1 ][line width=0.75]    (128,91.6) -- (171,55.6) ;
\draw  [fill={rgb, 255:red, 0; green, 0; blue, 0 }  ,fill opacity=1 ] (167,55.6) .. controls (167,53.39) and (168.79,51.6) .. (171,51.6) .. controls (173.21,51.6) and (175,53.39) .. (175,55.6) .. controls (175,57.81) and (173.21,59.6) .. (171,59.6) .. controls (168.79,59.6) and (167,57.81) .. (167,55.6) -- cycle ;
\draw  [fill={rgb, 255:red, 0; green, 0; blue, 0 }  ,fill opacity=1 ] (167,115.6) .. controls (167,113.39) and (168.79,111.6) .. (171,111.6) .. controls (173.21,111.6) and (175,113.39) .. (175,115.6) .. controls (175,117.81) and (173.21,119.6) .. (171,119.6) .. controls (168.79,119.6) and (167,117.81) .. (167,115.6) -- cycle ;
\draw  [fill={rgb, 255:red, 0; green, 0; blue, 0 }  ,fill opacity=1 ] (124,92.6) .. controls (124,90.39) and (125.79,88.6) .. (128,88.6) .. controls (130.21,88.6) and (132,90.39) .. (132,92.6) .. controls (132,94.81) and (130.21,96.6) .. (128,96.6) .. controls (125.79,96.6) and (124,94.81) .. (124,92.6) -- cycle ;
\draw [color={rgb, 255:red, 0; green, 0; blue, 0 }  ,draw opacity=1 ][line width=0.75]    (89.33,120) -- (128,92.6) ;
\draw  [fill={rgb, 255:red, 0; green, 0; blue, 0 }  ,fill opacity=1 ] (85.33,120) .. controls (85.33,117.79) and (87.12,116) .. (89.33,116) .. controls (91.54,116) and (93.33,117.79) .. (93.33,120) .. controls (93.33,122.21) and (91.54,124) .. (89.33,124) .. controls (87.12,124) and (85.33,122.21) .. (85.33,120) -- cycle ;
\draw  [fill={rgb, 255:red, 0; green, 0; blue, 0 }  ,fill opacity=1 ] (86.33,58) .. controls (86.33,55.79) and (88.12,54) .. (90.33,54) .. controls (92.54,54) and (94.33,55.79) .. (94.33,58) .. controls (94.33,60.21) and (92.54,62) .. (90.33,62) .. controls (88.12,62) and (86.33,60.21) .. (86.33,58) -- cycle ;
\draw [color={rgb, 255:red, 0; green, 0; blue, 0 }  ,draw opacity=1 ][line width=0.75]    (90.33,58) -- (128,92.6) ;
\draw [color={rgb, 255:red, 0; green, 0; blue, 0 }  ,draw opacity=1 ][line width=0.75]    (128,92.6) -- (171,115.6) ;
\draw [color={rgb, 255:red, 0; green, 0; blue, 0 }  ,draw opacity=1 ][line width=0.75]    (171,115.6) -- (130,146.6) ;
\draw  [fill={rgb, 255:red, 0; green, 0; blue, 0 }  ,fill opacity=1 ] (126,146.6) .. controls (126,144.39) and (127.79,142.6) .. (130,142.6) .. controls (132.21,142.6) and (134,144.39) .. (134,146.6) .. controls (134,148.81) and (132.21,150.6) .. (130,150.6) .. controls (127.79,150.6) and (126,148.81) .. (126,146.6) -- cycle ;
\draw   (374.33,105.5) .. controls (374.33,70.43) and (397.17,42) .. (425.33,42) .. controls (453.5,42) and (476.33,70.43) .. (476.33,105.5) .. controls (476.33,140.57) and (453.5,169) .. (425.33,169) .. controls (397.17,169) and (374.33,140.57) .. (374.33,105.5) -- cycle ;
\draw   (44.33,96.55) .. controls (44.33,58.19) and (85.67,27.1) .. (136.67,27.1) .. controls (187.66,27.1) and (229,58.19) .. (229,96.55) .. controls (229,134.91) and (187.66,166) .. (136.67,166) .. controls (85.67,166) and (44.33,134.91) .. (44.33,96.55) -- cycle ;
\draw [color={rgb, 255:red, 0; green, 0; blue, 0 }  ,draw opacity=1 ][line width=0.75]  [dash pattern={on 4.5pt off 4.5pt}]  (121.33,263.67) -- (272,165.6) ;
\draw [color={rgb, 255:red, 0; green, 0; blue, 0 }  ,draw opacity=1 ][line width=0.75]  [dash pattern={on 4.5pt off 4.5pt}]  (272,165.6) -- (307.33,200) ;
\draw [color={rgb, 255:red, 0; green, 0; blue, 0 }  ,draw opacity=1 ][line width=0.75]  [dash pattern={on 4.5pt off 4.5pt}]  (210.33,263.67) -- (272,165.6) ;
\draw [color={rgb, 255:red, 0; green, 0; blue, 0 }  ,draw opacity=1 ][line width=0.75]  [dash pattern={on 4.5pt off 4.5pt}]  (341.33,247) -- (367.33,293.67) ;
\draw [color={rgb, 255:red, 0; green, 0; blue, 0 }  ,draw opacity=1 ][line width=0.75]  [dash pattern={on 4.5pt off 4.5pt}]  (307.33,200) -- (341.33,247) ;
\draw [color={rgb, 255:red, 0; green, 0; blue, 0 }  ,draw opacity=1 ][line width=0.75]  [dash pattern={on 4.5pt off 4.5pt}]  (307.33,200) -- (276.33,272.67) ;
\draw   (121.33,263.67) -- (151,337.67) -- (91.67,337.67) -- cycle ;
\draw   (276.33,272.67) -- (296.5,315.67) -- (256.17,315.67) -- cycle ;
\draw   (210.33,263.67) -- (240,337.67) -- (180.67,337.67) -- cycle ;
\draw [color={rgb, 255:red, 0; green, 0; blue, 0 }  ,draw opacity=1 ][line width=0.75]  [dash pattern={on 4.5pt off 4.5pt}]  (312.33,329.67) -- (341.33,247) ;
\draw   (312.33,327) -- (338.42,392.33) -- (286.25,392.33) -- cycle ;
\draw  [dash pattern={on 4.5pt off 4.5pt}]  (171,115.6) .. controls (218.33,114.67) and (266.33,140.67) .. (272,165.6) ;
\draw  [dash pattern={on 4.5pt off 4.5pt}]  (405,103.6) .. controls (353.33,132.67) and (418.33,245.67) .. (367.33,293.67) ;
\draw  [fill={rgb, 255:red, 255; green, 255; blue, 255 }  ,fill opacity=1 ] (363.33,293.67) .. controls (363.33,291.46) and (365.12,289.67) .. (367.33,289.67) .. controls (369.54,289.67) and (371.33,291.46) .. (371.33,293.67) .. controls (371.33,295.88) and (369.54,297.67) .. (367.33,297.67) .. controls (365.12,297.67) and (363.33,295.88) .. (363.33,293.67) -- cycle ;
\draw  [fill={rgb, 255:red, 255; green, 255; blue, 255 }  ,fill opacity=1 ] (337.33,247) .. controls (337.33,244.79) and (339.12,243) .. (341.33,243) .. controls (343.54,243) and (345.33,244.79) .. (345.33,247) .. controls (345.33,249.21) and (343.54,251) .. (341.33,251) .. controls (339.12,251) and (337.33,249.21) .. (337.33,247) -- cycle ;
\draw  [fill={rgb, 255:red, 255; green, 255; blue, 255 }  ,fill opacity=1 ] (303.33,200) .. controls (303.33,197.79) and (305.12,196) .. (307.33,196) .. controls (309.54,196) and (311.33,197.79) .. (311.33,200) .. controls (311.33,202.21) and (309.54,204) .. (307.33,204) .. controls (305.12,204) and (303.33,202.21) .. (303.33,200) -- cycle ;
\draw  [fill={rgb, 255:red, 255; green, 255; blue, 255 }  ,fill opacity=1 ] (268,165.6) .. controls (268,163.39) and (269.79,161.6) .. (272,161.6) .. controls (274.21,161.6) and (276,163.39) .. (276,165.6) .. controls (276,167.81) and (274.21,169.6) .. (272,169.6) .. controls (269.79,169.6) and (268,167.81) .. (268,165.6) -- cycle ;
\draw  [color={rgb, 255:red, 255; green, 0; blue, 0 }  ,draw opacity=1 ] (468.33,49) .. controls (525.33,106) and (476.33,241) .. (433.33,297) .. controls (390.33,353) and (325.33,279) .. (316.33,254) .. controls (307.33,229) and (288.33,185) .. (234.33,172) .. controls (180.33,159) and (164.33,189) .. (138.33,191) .. controls (112.33,193) and (71.33,189) .. (49.33,167) .. controls (27.33,145) and (-3.67,52) .. (85.33,25) .. controls (174.33,-2) and (411.33,-8) .. (468.33,49) -- cycle ;
\draw   (115.33,251) .. controls (148.67,227) and (219.33,219) .. (270.33,247) .. controls (321.33,275) and (377.33,352) .. (351.33,395) .. controls (325.33,438) and (136.33,415) .. (94.33,365) .. controls (52.33,315) and (82,275) .. (115.33,251) -- cycle ;
\draw  [color={rgb, 255:red, 0; green, 0; blue, 255 }  ,draw opacity=1 ] (277.33,139) .. controls (362.33,140) and (453.33,360) .. (427.33,403) .. controls (401.33,446) and (87.33,448) .. (56.33,359) .. controls (25.33,270) and (192.33,138) .. (277.33,139) -- cycle ;
\draw  [dash pattern={on 4.5pt off 4.5pt}]  (492.33,281) .. controls (472.33,328) and (412.33,385) .. (313.33,371.67) ;
\draw  [fill={rgb, 255:red, 255; green, 255; blue, 255 }  ,fill opacity=1 ] (488.33,281) .. controls (488.33,278.79) and (490.12,277) .. (492.33,277) .. controls (494.54,277) and (496.33,278.79) .. (496.33,281) .. controls (496.33,283.21) and (494.54,285) .. (492.33,285) .. controls (490.12,285) and (488.33,283.21) .. (488.33,281) -- cycle ;
\draw  [fill={rgb, 255:red, 255; green, 255; blue, 255 }  ,fill opacity=1 ] (309.33,371.67) .. controls (309.33,369.46) and (311.12,367.67) .. (313.33,367.67) .. controls (315.54,367.67) and (317.33,369.46) .. (317.33,371.67) .. controls (317.33,373.88) and (315.54,375.67) .. (313.33,375.67) .. controls (311.12,375.67) and (309.33,373.88) .. (309.33,371.67) -- cycle ;
\draw    (501,272) -- (518.83,256.32) ;
\draw [shift={(520.33,255)}, rotate = 138.67] [color={rgb, 255:red, 0; green, 0; blue, 0 }  ][line width=0.75]    (10.93,-3.29) .. controls (6.95,-1.4) and (3.31,-0.3) .. (0,0) .. controls (3.31,0.3) and (6.95,1.4) .. (10.93,3.29)   ;
\draw    (505,290) -- (520.64,299.93) ;
\draw [shift={(522.33,301)}, rotate = 212.4] [color={rgb, 255:red, 0; green, 0; blue, 0 }  ][line width=0.75]    (10.93,-3.29) .. controls (6.95,-1.4) and (3.31,-0.3) .. (0,0) .. controls (3.31,0.3) and (6.95,1.4) .. (10.93,3.29)   ;

\draw (170,126) node [anchor=north west][inner sep=0.75pt]    {$u$};
\draw (395,78.4) node [anchor=north west][inner sep=0.75pt]    {$v$};
\draw (281,148.4) node [anchor=north west][inner sep=0.75pt]    {$x_{1}$};
\draw (318,179.4) node [anchor=north west][inner sep=0.75pt]    {$x_{2}$};
\draw (347,223) node [anchor=north west][inner sep=0.75pt]    {$x_{3}$};
\draw (381,287.4) node [anchor=north west][inner sep=0.75pt]    {$x_{4}$};
\draw (502,137.4) node [anchor=north west][inner sep=0.75pt]  [color={rgb, 255:red, 255; green, 0; blue, 0 }  ,opacity=1 ]  {$T^\new_u$};
\draw (225,48.4) node [anchor=north west][inner sep=0.75pt]    {$T_{u}$};
\draw (224,386.4) node [anchor=north west][inner sep=0.75pt]  [color={rgb, 255:red, 0; green, 0; blue, 0 }  ,opacity=1 ]  {$W'$};
\draw (35,413.4) node [anchor=north west][inner sep=0.75pt]  [color={rgb, 255:red, 0; green, 0; blue, 255 }  ,opacity=1 ]  {DFS tree};
\draw (294,372.4) node [anchor=north west][inner sep=0.75pt]    {$y$};
\draw (488.33,290.4) node [anchor=north west][inner sep=0.75pt]    {$z$};
\draw (530,295.4) node [anchor=north west][inner sep=0.75pt]    {or $z \in V(T^\new)$ and $\deg_{T^\new}(z) =4$};
\draw (531,244.4) node [anchor=north west][inner sep=0.75pt]    {either $z \in V(T^\new_u)$};

\end{tikzpicture}

\end{figure}
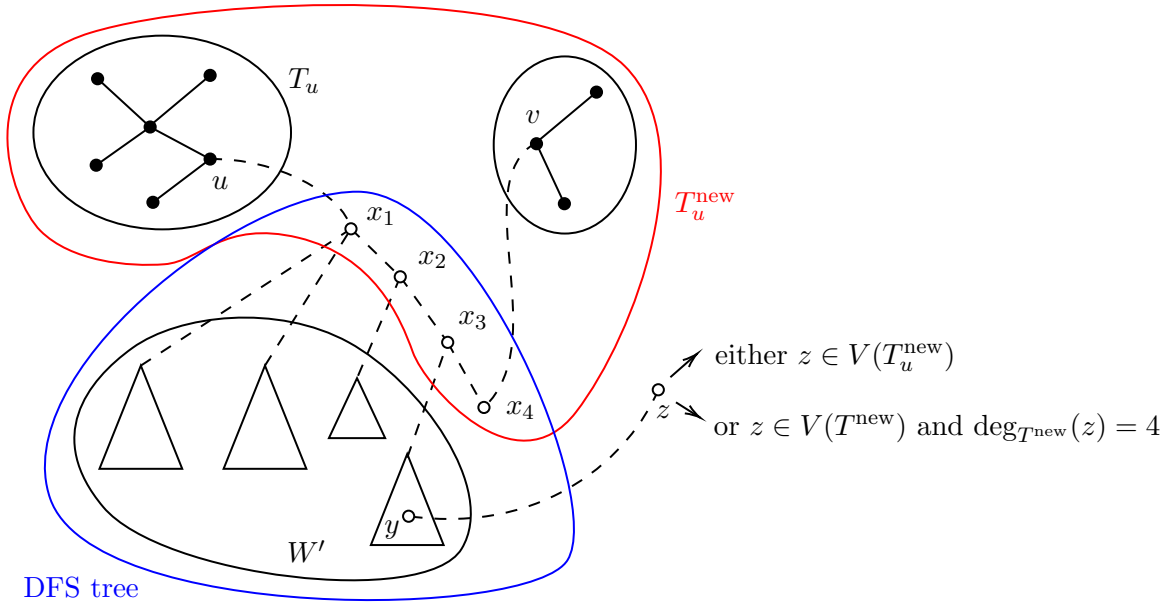

%% file: hierarchy.tex
\section{Optimal Low-Degree Hierarchy}\label{sec:hierarchy}

In this section, we recall the definition of a \emph{low-degree hierarchy} from Definition 5.1  of \cite{LS22} and construct such a hierarchy with optimal parameters in near-linear time in \Cref{thm:main-hierarchy}. 
In \Cref{sec:oracle}, we plug this hierarchy into the known framework to get an improved connectivity oracle.

\begin{definition}[Low-Degree Hierarchy]\label{def:LowDegreeHierarchy}
Let $G$ be a connected undirected graph.
A $(p,s)$-low-degree hierarchy of $G$, with depth $p$ and degree parameter $s$, is a sequence $\{(\calC_i, \calT_i)\}_{i=1}^p$ satisfying the following properties:

\begin{itemize}
    \item[(1)] 
    At each level $i \in [1,p]$, $\calC_i$ is a collection of pairwise vertex-disjoint, vertex-induced subgraphs of $G$, and no edge of $G$ joins two distinct members of $\calC_i$.
    In particular, $\calC_p = \{G\}$.
    
    \item[(2)] The family $\bigcup_{i=1}^p \calC_i$ is laminar: for every $i\in[1,p-1]$ and every $\gamma \in \calC_i$, there is a unique $\gamma' \in \calC_{i+1}$ such that $V(\gamma) \subseteq V(\gamma')$.
    We call $\gamma'$ the parent of $\gamma$, and $\gamma$ a child of $\gamma'$.
\end{itemize}
For each level $i \in [1,p]$ and component $\gamma \in \calC_i$, the terminal set of $\gamma$, denoted by $U(\gamma)$, consists of the vertices in $V(\gamma)$ that belong to none of its children.
The set $U(\gamma)$ may be empty; in particular, $U(\gamma)=V(\gamma)$ for every $\gamma \in \calC_1$.
The terminal set at level $i$ is $U_i := \bigcup_{\gamma \in \calC_i} U(\gamma)$.
\begin{itemize}
    \item[(3)] At each level $i \in [1,p]$, $\calT_i$ is a Steiner forest for $U_i$ with respect to $G-U_{>i}$ and has maximum degree at most $s$, where $U_{>i} := \bigcup_{j=i+1}^p U_j$ and $U_{>p}:=\emptyset$.
\end{itemize}
\end{definition}

Now, we prove the following theorem.

\begin{restatable}{theorem}{thmhierarchy}
\label{thm:main-hierarchy}
        A deterministic algorithm constructs a low-degree hierarchy of depth $O(\log n)$ and degree parameter $4$ in $O(m\alpha(n)\log n)$ time.
\end{restatable}

This construction strictly subsumes all prior low-degree hierarchy constructions from \cite{DP20,LS22,li2024congestion}
and also has optimal parameters; any hierarchy with degree $O(1)$ needs $\Omega(\log n)$ depth. See \Cref{thm:lower} for the proof of this trade-off.

Below, we describe the construction.

\paragraph{Construction.}
    
Set $X_1:=V$.
For each $i\geq1$ with $X_i\neq\emptyset$, apply the algorithm of \Cref{thm:main-decomp} to obtain $(X_{i+1},T_i)=\CTP(G,X_i)$.
By Property~\ref{prop:decomp-size}, after at most $1+\lceil\log_2 n\rceil$ calls there is an index $i^\star$ such that $X_{i^\star+1}=\emptyset$.
Set $p:=i^\star$ and define the $(p,4)$-low-degree hierarchy as follows:
\begin{itemize}
    \item $\calT_i:=T_i$ for each $i\in[1,p]$.
    \item $\calC_i$ is the collection of connected components of $G-X_{>i}$, where $X_{>i}:=\bigcup_{j=i+1}^p X_j$ and $X_{>p}:=\emptyset$.
\end{itemize}

\paragraph{Analysis.}
Properties (1) and (2) of \Cref{def:LowDegreeHierarchy} follow directly from the definition of $\calC_i$.
By Property~\ref{prop:decomp-forest}, every forest $\calT_i$ has maximum degree at most $4$.
The terminals at level $i$ are precisely the vertices in components of $\calC_i$ that belong to no lower-level component.
Thus, $U_i=V(\calC_i)\cap X_i$, so $U_i\subseteq X_i$ and $U_{>i}=X_{>i}$ for every $i\in[1,p]$.
Property~\ref{prop:decomp-forest} also shows that $\calT_i=T_i$ is a Steiner forest for $X_i$ with respect to $G-X_{i+1}$.
Since $U_i\subseteq X_i$ and $X_{i+1}\subseteq X_{>i}=U_{>i}$, the observation below implies that $\calT_i$ is a Steiner forest for $U_i$ with respect to $G-U_{>i}$.

\begin{observation}
    Let $X,U \subseteq V$, and let $T \subseteq G$ be a Steiner forest for $U$ with respect to $G-X$.
    \begin{enumerate}
        \item For every $U' \subseteq U$, $T$ is also a Steiner forest for $U'$ with respect to $G-X$.
        
        \item For every $X' \supseteq X$, $T$ is also a Steiner forest for $U$ with respect to $G-X'$.
    \end{enumerate}
\end{observation}

\begin{proof}
    For the first part, a component of $T$ containing $U\cap V(C)$ also contains $U'\cap V(C)\subseteq U\cap V(C)$.
    The second part follows because every connected component of $G-X'$ is contained in a connected component of $G-X$.
\end{proof}

%% file: connctivity-oracle.tex
\section{Connectivity Oracles under Vertex Failures}
\label{sec:oracle}

We first combine our hierarchy with the state-of-the-art oracle of \cite{LS22} to prove \Cref{thm:main-oracle}.
We then give a self-contained description of a simpler oracle based on \cite{DP20}, at the cost of weaker bounds.

\subsection{Connectivity Oracle of \cite{LS22} Combined with Our Low-Degree Hierarchy}\label{sec:combine-LS}

We use the framework of \cite{LS22} as a black box.

\begin{proof}[Proof of \Cref{thm:main-oracle}]
    First, construct a $(p,4)$-low-degree hierarchy using \Cref{thm:main-hierarchy}, where $p=O(\log n)$.
    This takes $O(m\alpha(n)\log n)$ time and uses $O(m\log n)$ space.
    By Lemma 6.14 in the arXiv version of \cite{LS22}, the remaining data structures can be constructed in $O(d_\star p m\log^2 n)=O(d_\star m\log^3 n)$ time using $O(pm\log^2 n)=O(m\log^3 n)$ space.
    Thus, the total preprocessing time and space are $O(d_\star m\log^3 n)$ and $O(m\log^3 n)$, respectively.

    Theorem 7.2 in the arXiv version of \cite{LS22} bounds the update time for a set $D$ of $d\leq d_\star$ failed vertices by
    \[
    O\bigl(\Delta^2p^2d^2\log n+\Delta p^2d^2\log^4(\Delta p d)\log n\bigr)
    =O\bigl(d^2\log^3 n\log^4(d\log n)\bigr),
    \]
    where $(p,\Delta)=(O(\log n),4)$.

    Finally, the procedure in Section 7.3 of the arXiv version of \cite{LS22}, which is essentially the query procedure of \cite{DP20}, has query time $O(d)$.
\end{proof}

\subsection{Connectivity Oracle of \cite{DP20} Combined with Our Low-Degree Hierarchy}\label{sec:combine-DP}

For completeness, we present a simple connectivity oracle based on \cite{DP20}.
We use elementary data structures, incurring only a polylogarithmic overhead.

\begin{restatable}{theorem}{thmDP}
\label{thm:oracleDP}
    Given a connected undirected graph $G$ with $n$ vertices and $m$ edges and an upper bound $d_\star$ on the failure-set size $d$, there is a simple deterministic $d_\star$-vertex-failure connectivity oracle with the following guarantees:
    \begin{enumerate}
        \item $O(d_\star m \log^2 n)$ space.
        \item $O(d_\star m \log^2 n)$ preprocessing time.
        \item $O(d^2 (\log^2 n) (d_\star^2 + \log^2 n))$ update time.
        \item $O(d + \log\log n)$ query time.
    \end{enumerate}
\end{restatable}

\subsubsection{Some Necessary Structures}

\textbf{Simple Two-Dimensional Orthogonal Range Reporting.}
Given a set of $N$ points on the grid $[M]\times[M]$, for an integer $M$, there is a data structure of size $O(N\log N)$ that reports the $k$ points in any rectangle $[x,y]\times[w,z]$ in $O(\log^2 N+k)$ time.

\medskip
\noindent
\textbf{Simple Edge-Failure Euler-Tour Structure.}
Given a forest $F$ of a graph $G$, there is a data structure of size $O(m\log n)$ with the following guarantees.
Let $Z$ be a set of deleted edges, $d$ of which belong to $F$ and $d'$ of which do not.
The forest $F'=F-Z$ has at most $2d$ affected subtrees, meaning components created within trees of $F$ that contain deleted forest edges.
In $O(d^2\log^2 n+d')$ time, the data structure reports which pairs of affected subtrees are joined by an edge of $E(G)-Z$; it also identifies the affected subtree containing a query vertex $u$ in $O(\log d)$ time.
\begin{itemize}
    \item \textbf{Constructing the Structure.}
    Consider an arbitrary root for $F$ and order its nodes by DFS preorder (equivalently, by their first occurrence in an Euler tour). 
    Map the nodes to integers in this order.
    For each pair $(T_1,T_2)$ of not necessarily distinct trees of $F$, construct a two-dimensional orthogonal range-reporting structure by adding the points corresponding to $(x,y)$ and $(y,x)$ for every edge $\{x,y\}\in E(G)$ between the two trees.
    To save space, construct this structure only for pairs joined by an edge of $G$.
    Assign an identifier $\operatorname{id}(T)$ to each component $T$ of $F$, and store $\operatorname{id}(T_1)n+\operatorname{id}(T_2)$ in a set $S$ for every pair $(T_1,T_2)$ joined by an edge.

    \item \textbf{Reconnecting Affected Subtrees.} Removing $q$ edges from a tree partitions its vertex ordering into at most $2q+1$ intervals.
    For every pair of intervals corresponding to affected subtrees, first query $S$ in $O(\log n)$ time to determine whether the relevant range-reporting structure exists. If it does, query that structure for an edge joining the two subtrees; otherwise, no such edge exists in $G$.
    Here, we stop the range reporting procedure once it finds at least one edge outside $Z$.

    \item \textbf{Locating the Affected Subtree Containing $u$.} Binary search over the sorted left endpoints of the affected intervals identifies the affected subtree containing $u$ in $O(\log d)$ time.
\end{itemize}

\medskip
\noindent
\textbf{Final Euler-Tour Structure.}
If the forest $F$ has $O(1)$ maximum degree, the previous Euler-Tour structure can handle vertex and edge failures as well.
We denote this final Euler-Tour structure by $\textbf{ET}(G,F)$.
More precisely, given a failed vertex-set $D$ and a failed edge-set $Z$, the nonforest edges incident to $D$ are only deleted implicitly, and they are irrelevant to the range query.
In other words, we do not perform a range query on the singleton nodes in $D$, and it is not necessary to iterate over all edges incident to $D$ and remove them from the graph explicitly (there may be $\Theta(m)$ such edges).
So, $\textbf{ET}(G,F)$ has $O(d^2 \log^2 n + d')$ update time and $O(\log d)$ query time, where $d$ is the total number of vertices and edges deleted from $F$ and $d'$ is the total number of vertices and edges removed from outside $F$.

\medskip
\noindent
\textbf{$d_\star$-Adjacency Edges.}
Let $L=(v_1,v_2,\ldots,v_r)$ be a list of vertices.
The set $\Lambda_{d_\star}(L)$ contains an edge between every two vertices at distance at most $d_\star+1$ in $L$:
\[
\Lambda_{d_\star}(L)=\bigl\{\{v_i,v_j\}:1\leq i<j\leq r \text{ and } j-i\leq d_\star+1\bigr\}.
\]

\subsubsection{Constructing the Oracle}

First, use \Cref{thm:main-hierarchy} to compute a $(p,s)$-low-degree hierarchy $\{(\calC_i,\calT_i)\}_{i=1}^p$, where $p=O(\log n)$ and $s=4$.
For brevity, let $\calC:=\bigcup_{i=1}^p\calC_i$ and $\calT:=\bigcup_{i=1}^p\calT_i$.
For each $i\in[1,p]$ and $\gamma_i\in\calC_i$, let $\gamma_{i+1},\ldots,\gamma_p$ denote the ancestors of $\gamma_i$.
Whenever $U(\gamma_i)\neq\emptyset$, let $\tau(\gamma_i)\in\calT_i$ be the unique tree containing all terminals in $U(\gamma_i)$; no such choice is needed when $U(\gamma_i)=\emptyset$.
For each $j\in[i+1,p]$, let $A(\gamma_i,\gamma_j)$ be the list of terminals in $U(\gamma_j)$ adjacent to at least one vertex of $V(\gamma_i)$, ordered according to the same Euler tour of $\tau(\gamma_j)$.
If $U(\gamma_j)=\emptyset$, this list is empty and no choice of $\tau(\gamma_j)$ is needed.
Define $A(\gamma_i)$ as the concatenation of $A(\gamma_i,\gamma_{i+1}),\ldots,A(\gamma_i,\gamma_p)$.
We identify each entry of $A(\gamma_i)$ with the corresponding terminal copy in its tree of $\calT$.

\medskip
\noindent
\textbf{Auxiliary Graph.}
Define a multigraph $H$ whose vertices are the vertices of the forests in $\calT$.
An original vertex $u\in V(G)$ may appear in several trees of $\calT$; these appearances are represented by distinct copies in $H$, and we have $|V(H)|\leq pn$.
The edge set $E(H)$ consists of the following edges; 
for each $\{u,v\} \in E(G)$, we add an original edge in $H$ between the terminal copies of $u$ and $v$.
For each component $\gamma\in\calC$, we add the edges of $\Lambda_{d_\star}(A(\gamma))$ to $H$ as artificial edges labeled by $\gamma$.
For each tree $\tau \in \calT$, and each edge of $\tau$, we add an edge labeled by $\tau$ to $H$ between the corresponding copies of the endpoints.
Thus, $H$ may have parallel edges with different labels, but each original edge contributes to only $O(p)$ lists $A(\gamma)$.
Thus, $|E(H)| \leq m + p(d_\star+1)m + p(n-1) = O(d_\star m \log n)$.
Finally, construct the Euler-tour data structure $\textbf{ET}(H,\calT)$ that dominates the space and preprocessing time $O(|E(H)| \log |V(H)|) = O(d_\star m \log^2 n)$.

\subsubsection{Updating the Oracle After Vertex Failure}

Let $D\subseteq V$ be a set of $d\leq d_\star$ failed vertices.

\medskip
\noindent
\textbf{Affected Data Structures.}
We call a component $\gamma\in\calC$ affected if $V(\gamma)\cap D\neq\emptyset$.
When an affected component $\gamma$ also satisfies $U(\gamma)\neq\emptyset$, we call its corresponding tree $\tau(\gamma)$ affected as well.
Moreover, if any tree in $\calT$ contains a copy of a failed vertex, we also call it an affected tree.\footnote{We need to explicitly consider these trees affected as well since in our low-degree hierarchy, a tree $\tau(\gamma)$ can contain nodes outside $V(\gamma)$. But the bound $O(dp)$ on the number of affected trees will not change.}
Remove every copy of a vertex in $D$ and its incident edges from each affected tree, breaking it into affected subtrees, and mark the corresponding entries in $A(\gamma)$ as dead (more precisely, for implementation, the list remains fixed; an entry is treated as dead if its underlying vertex belongs to $D$).
There are $O(dp)$ affected trees and subtrees in total, because each original vertex has at most $p$ copies.
For convenience, from now on, we also refer to the entire affected trees that have no failed vertices as affected subtrees.

\medskip
\noindent
\textbf{Removing Invalid Artificial Edges.}
There are $O(pd)$ breaks (as well as these many intervals) in each affected component $\gamma$, including the break between any two consecutive concatenated lists $A(\gamma, \gamma_{j})$ and $A(\gamma, \gamma_{j+1})$, as well as the break caused when deletion of forest vertices/edges partitions an ancestor tree’s preorder into intervals.
We ignore every edge of $\Lambda_{d_\star}(A(\gamma))$ whose endpoints lie in distinct intervals.
Since $O(d_\star^2)$ edges cross each boundary and there are $O(pd)$ affected components, the number of edges of $H$ that must be ignored is bounded by $O(p^2d^2d_\star^2)$.
Let $H'$ be the graph obtained by removing these edges from $H$.

\medskip
\noindent
\textbf{Reconnecting Affected Subtrees.}
Let $R$ be a graph whose vertices correspond to the affected subtrees.
Using $\textbf{ET}(H,\calT)$, add an edge between $t,t'\in V(R)$ if an edge of $H'$ joins the corresponding subtrees.
The construction time is $O((pd)^2\log^2 n+p^2d^2d_\star^2)$: respectively, $O(pd)$ failed copies and incident forest edges, and $O(p^2d^2d_\star^2)$ nonforest edges are removed from $H$.
Note that there might be a tree $\tau \in \calT$ marked as affected that has no failed vertex.
So, we need to query $\textbf{ET}(H,\calT)$ for those trees as well, but the bound $O((pd)^2 \log^2 n)$ still remains valid since there are at most $O(pd)$ many such affected trees with no failed vertex.
Finally, compute the connected components of $R$ in $O((pd)^2)$ time.

\subsubsection{Answering Queries}

If either query vertex belongs to $D$, report that they are disconnected.
Otherwise, for $u,v\notin D$, perform the following steps to determine whether they are connected in $G-D$.

\medskip
\noindent
\textbf{Step 1.}
Find the components in $\calC$ that contain $u$ and $v$ as terminals.
Let $\gamma(u)$ and $\gamma(v)$ be those components. 
If $\gamma(u)$ is unaffected, let $\hat{\gamma}(u)$ be its highest unaffected ancestor.
We can find $\hat{\gamma}(u)$ in $O(\log\log n)$ time by binary-searching the ancestors of $\gamma(u)$.
Define $\hat{\gamma}(v)$ analogously.
If both $\hat{\gamma}(u)$ and $\hat{\gamma}(v)$ exist and are equal, report that $u$ and $v$ are connected and stop.

\medskip
\noindent
\textbf{Step 2.}
Find vertices $u'$ and $v'$ in affected subtrees that are connected to $u$ and $v$, respectively.
If $\gamma(u)$ is affected, set $u':=u$ (terminal copy of $u$ in $H$).
Otherwise, scan $A(\hat{\gamma}(u))$ for a vertex $u'\notin D$ adjacent to $V(\hat{\gamma}(u))$.
The vertices $u$ and $u'$ are connected because $\hat{\gamma}(u)$ is unaffected, and $u'$ belongs to an affected subtree because every proper ancestor of $\hat{\gamma}(u)$ is affected.
It takes $O(d)$ time to find $u'$, since it suffices to scan at most $d+1>|D|$ entries of $A(\hat{\gamma}(u))$.
If $u'$ does not exist, then no surviving edge leaves $V(\hat{\gamma}(u))$.
Define $v'$ analogously, and if either $u'$ or $v'$ does not exist, report that $u$ and $v$ are disconnected and stop.

\medskip
\noindent
\textbf{Step 3.}
Use $\textbf{ET}(H,\calT)$ to find the affected subtrees containing $u'$ and $v'$ in $O(\log (pd))$ time, and let $t_{u'}$ and $t_{v'}$ be the corresponding vertices of $R$.
Report that $u$ and $v$ are connected if and only if $t_{u'}$ and $t_{v'}$ lie in the same connected component of $R$.
This test takes $O(1)$ time.

\paragraph{Correctness.}
Every path in $R$ corresponds to a path in $G-D$; surviving forest and original edges represent valid edges of $G-D$, and an artificial edge labeled by an unaffected component $\gamma$ can be replaced by a path through $G[V(\gamma)]$.
Artificial edges with an affected label whose endpoints lie in distinct affected subtrees are ignored in the range queries.

Conversely, map the vertices of a path $P \subseteq G-D$ to their terminal copies and contract every maximal subpath lying in a highest unaffected component $\gamma$.
Consecutive affected subtrees are then either joined by an original edge of $P$, or separated by an unaffected component $\gamma$.
In the latter case, their external boundary terminals belong to $A(\gamma)$ and are connected by a path in $\Lambda_{d_\star}(A(\gamma))-D$.
Since every proper ancestor of $\gamma$ is affected, all surviving terminals on this artificial path belong to affected subtrees.
Thus the corresponding vertices are connected in $R$.

Hence, for any two surviving terminal copies $a, b$ contained in affected subtrees, their $R$-vertices are connected in $R$ if and only if the corresponding original vertices are connected in $G-D$.

%% file: appendix-lower.tex
\section{Optimality of Our Low-Degree Hierarchy}

In this section, we show that our construction of the low-degree hierarchy in \Cref{sec:hierarchy} is optimal.

\begin{lemma}\label{lem:lower}
    For any $k \geq 2$ and $\ell \geq 1$, there exists a graph $G$ with $ n = \Theta(k^\ell)$ nodes such that \textbf{any} low-degree hierarchy of $G$ either has max-degree $\geq k$ or has depth $\geq \ell = \Omega(\log_k n)$.
\end{lemma}

\noindent
Before we prove this lemma, we show how this implies that our hierarchy is optimal.

\begin{restatable}{theorem}{lowerbound}
\label{thm:lower}
    Assume that $s,p : \mathbb{N} \rightarrow \mathbb{N}$ are two non-decreasing functions.
    If there exists an algorithm (randomized or deterministic) that constructs a $(p(n),s(n))$-low-degree hierarchy on input graphs of $n$ nodes, then we must have that $(s(n)+1)^{p(n) + 1} =  \Omega(n)$.
    Accordingly, $ p(n) \cdot \log (s(n)) = \Omega(\log n)$.
\end{restatable}

\begin{proof}
    By \Cref{lem:lower} for $k = s(n) + 1$ and $\ell = p(n) + 1$, there exist a graph $G$ with $n' = \Theta((s(n)+1)^{p(n) + 1}) \leq C \cdot (s(n)+1)^{p(n) + 1}$ (for some pure constant $C$) nodes such that any low-degree hierarchy of $G$ either has maximum degree at least $s(n) + 1$ or has depth at least $p(n)+1$.
    Now, we show that $n' > n$.
    If $n' \leq n$, by running the existential algorithm on $G$, we achieve a low-degree hierarchy of $G$ with max-degree $\leq s(n') \leq s(n)$ and depth $\leq p(n') \leq p(n)$.
    This is clearly a contradiction.
    So, we have
    $$ n \leq n' \leq C \cdot (s(n)+1)^{p(n) + 1}. \qedhere $$
\end{proof}

This theorem shows that our construction of the low-degree hierarchy is optimal because any algorithm that constructs a low-degree hierarchy with an $O(1)$ degree parameter must have a depth of at least $\Omega(\log n)$.

\subsection{Proof of \Cref{lem:lower}}

We define $G$ inductively.
Let $G_0$ be a graph that consists of a single node $r_0$ called the \textbf{root} of $G_0$.
For any $i \geq 1$, we construct $G_i$ as follows;
\begin{enumerate}
    \item Consider a node $r_i$ as the root of $G_i$.

    \item Make $k$ copies of $G_{i-1}$ like $G_{i-1}^{(1)}, G_{i-1}^{(2)}, \ldots, G_{i-1}^{(k)}$.

    \item For each $j \in [1, k]$, connect $r_i$ to the root $r_{i-1}^{(j)}$ of $G_{i-1}^{(j)}$.
\end{enumerate}
The following figure illustrates $G_1$ and $G_2$ for $k=3$.

\input{fig-lower}

\noindent
Eventually, we define $G$ to be $G_\ell$.
It is straightforward to see that the number of nodes in $G$ is 
$ k^\ell + k^{\ell-1} + \cdots + 1 = \Theta(k^\ell)$.
Now, consider a $(p, s)$-low-degree hierarchy of $G$ with degree parameter $s < k$.
We show the depth of the hierarchy is at least $\ell$.
This directly follows from the claim below.

\begin{claim}
    For any $i \in [0, \ell]$, the low-degree hierarchy has depth at least $i$ and moreover $U_{>i}$ must contain at least one node from each copy of $G_i$ in the construction of $G$ (recall that there are $k^{\ell - i}$ copies of $G_i$ in total).
\end{claim}

\begin{proof}
    The claim is trivial for $i=0$ since $U_{>0}=V$.
    Now, assume that the claim holds for $i-1$; in particular, the hierarchy has depth at least $i-1$ and $U_{> i-1}$ is nonempty.
    Since $U_{> i-1}$ is non-empty, the hierarchy must have depth at least $i$.
    So, level $i$ exists and $U_{> i-1} = U_i \cup U_{> i}$.
    Consider one copy of $G_{i}$ in the construction of $G$ and consider the internal copies $G_{i-1}^{(1)},\ldots, G_{i-1}^{(k)}$ of $G_{i-1}$ inside $G_i$.
    According to the induction hypothesis $U_{>i-1}$ contains at least one node in each $G_{i-1}^{(j)}$ for $j \in [1, k]$.
    Now, for the sake of contradiction, assume that $U_{>i}$ does not contain any node in $G_i$.
    
    As a result, all the terminals in $U_{>i-1} \cap V(G_i)$ must belong to $U_i$    which means that each copy $G_{i-1}^{(j)}$ contains at least one node in $U_i$.
    Moreover, $\calT_i$ is a Steiner forest for $U_i$ w.r.t.~$G - U_{>i}$ and the entire $G_i$ is contained in one connected component of $G - U_{> i}$.
    We conclude that the terminals in $U_i \cap V(G_i)$ must be spanned by one connected component $T^\star$ (a tree) of $\calT_i$.
    Since $U_i$ has at least one node in each of the copies $G_{i-1}^{(j)}$, and these copies are connected to the rest of the graph only via the edge to the root $r_i$ of $G_i$, we conclude that the tree $T^\star$ must contain all the edges between $r_i$ and the roots of the copies $G_{i-1}^{(j)}$ for all $j \in [1, k]$.
    This is clearly in contradiction with the assumption that the maximum degree of $T^\star$ is $s < k$.
    Because this argument holds for every arbitrary copy of $G_i$, the proof is complete.
\end{proof}

%% file: fig-lower.tex
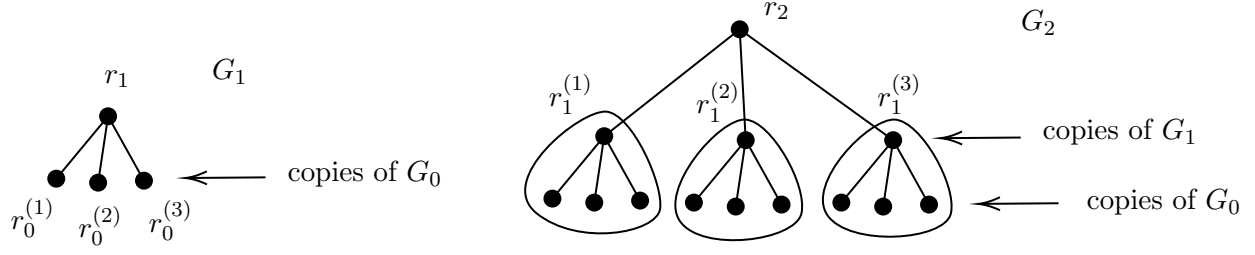
\begin{figure}[ht]
\caption{
An illustration of $G_1$ and $G_2$.}
  \centering

\tikzset{every picture/.style={line width=0.75pt}} 

\begin{tikzpicture}[x=0.75pt,y=0.75pt,yscale=-1,xscale=1]

\draw [color={rgb, 255:red, 0; green, 0; blue, 0 }  ,draw opacity=1 ][line width=0.75]    (459.33,165.33) -- (464.33,132) ;
\draw [color={rgb, 255:red, 0; green, 0; blue, 0 }  ,draw opacity=1 ][line width=0.75]    (464.33,132) -- (482.33,164.33) ;
\draw [color={rgb, 255:red, 0; green, 0; blue, 0 }  ,draw opacity=1 ][line width=0.75]    (438.33,163.33) -- (464.33,132) ;
\draw  [fill={rgb, 255:red, 0; green, 0; blue, 0 }  ,fill opacity=1 ] (460.33,132) .. controls (460.33,129.79) and (462.12,128) .. (464.33,128) .. controls (466.54,128) and (468.33,129.79) .. (468.33,132) .. controls (468.33,134.21) and (466.54,136) .. (464.33,136) .. controls (462.12,136) and (460.33,134.21) .. (460.33,132) -- cycle ;
\draw  [fill={rgb, 255:red, 0; green, 0; blue, 0 }  ,fill opacity=1 ] (478.33,164.33) .. controls (478.33,162.12) and (480.12,160.33) .. (482.33,160.33) .. controls (484.54,160.33) and (486.33,162.12) .. (486.33,164.33) .. controls (486.33,166.54) and (484.54,168.33) .. (482.33,168.33) .. controls (480.12,168.33) and (478.33,166.54) .. (478.33,164.33) -- cycle ;
\draw  [fill={rgb, 255:red, 0; green, 0; blue, 0 }  ,fill opacity=1 ] (455.33,165.33) .. controls (455.33,163.12) and (457.12,161.33) .. (459.33,161.33) .. controls (461.54,161.33) and (463.33,163.12) .. (463.33,165.33) .. controls (463.33,167.54) and (461.54,169.33) .. (459.33,169.33) .. controls (457.12,169.33) and (455.33,167.54) .. (455.33,165.33) -- cycle ;
\draw  [fill={rgb, 255:red, 0; green, 0; blue, 0 }  ,fill opacity=1 ] (434.33,163.33) .. controls (434.33,161.12) and (436.12,159.33) .. (438.33,159.33) .. controls (440.54,159.33) and (442.33,161.12) .. (442.33,163.33) .. controls (442.33,165.54) and (440.54,167.33) .. (438.33,167.33) .. controls (436.12,167.33) and (434.33,165.54) .. (434.33,163.33) -- cycle ;
\draw [color={rgb, 255:red, 0; green, 0; blue, 0 }  ,draw opacity=1 ][line width=0.75]    (385.33,165.33) -- (390.33,132) ;
\draw [color={rgb, 255:red, 0; green, 0; blue, 0 }  ,draw opacity=1 ][line width=0.75]    (390.33,132) -- (408.33,164.33) ;
\draw [color={rgb, 255:red, 0; green, 0; blue, 0 }  ,draw opacity=1 ][line width=0.75]    (364.33,163.33) -- (390.33,132) ;
\draw  [fill={rgb, 255:red, 0; green, 0; blue, 0 }  ,fill opacity=1 ] (386.33,132) .. controls (386.33,129.79) and (388.12,128) .. (390.33,128) .. controls (392.54,128) and (394.33,129.79) .. (394.33,132) .. controls (394.33,134.21) and (392.54,136) .. (390.33,136) .. controls (388.12,136) and (386.33,134.21) .. (386.33,132) -- cycle ;
\draw  [fill={rgb, 255:red, 0; green, 0; blue, 0 }  ,fill opacity=1 ] (404.33,164.33) .. controls (404.33,162.12) and (406.12,160.33) .. (408.33,160.33) .. controls (410.54,160.33) and (412.33,162.12) .. (412.33,164.33) .. controls (412.33,166.54) and (410.54,168.33) .. (408.33,168.33) .. controls (406.12,168.33) and (404.33,166.54) .. (404.33,164.33) -- cycle ;
\draw  [fill={rgb, 255:red, 0; green, 0; blue, 0 }  ,fill opacity=1 ] (381.33,165.33) .. controls (381.33,163.12) and (383.12,161.33) .. (385.33,161.33) .. controls (387.54,161.33) and (389.33,163.12) .. (389.33,165.33) .. controls (389.33,167.54) and (387.54,169.33) .. (385.33,169.33) .. controls (383.12,169.33) and (381.33,167.54) .. (381.33,165.33) -- cycle ;
\draw  [fill={rgb, 255:red, 0; green, 0; blue, 0 }  ,fill opacity=1 ] (360.33,163.33) .. controls (360.33,161.12) and (362.12,159.33) .. (364.33,159.33) .. controls (366.54,159.33) and (368.33,161.12) .. (368.33,163.33) .. controls (368.33,165.54) and (366.54,167.33) .. (364.33,167.33) .. controls (362.12,167.33) and (360.33,165.54) .. (360.33,163.33) -- cycle ;
\draw [color={rgb, 255:red, 0; green, 0; blue, 0 }  ,draw opacity=1 ][line width=0.75]    (314.33,163.33) -- (319.33,130) ;
\draw [color={rgb, 255:red, 0; green, 0; blue, 0 }  ,draw opacity=1 ][line width=0.75]    (319.33,130) -- (337.33,162.33) ;
\draw [color={rgb, 255:red, 0; green, 0; blue, 0 }  ,draw opacity=1 ][line width=0.75]    (293.33,161.33) -- (319.33,130) ;
\draw  [fill={rgb, 255:red, 0; green, 0; blue, 0 }  ,fill opacity=1 ] (315.33,130) .. controls (315.33,127.79) and (317.12,126) .. (319.33,126) .. controls (321.54,126) and (323.33,127.79) .. (323.33,130) .. controls (323.33,132.21) and (321.54,134) .. (319.33,134) .. controls (317.12,134) and (315.33,132.21) .. (315.33,130) -- cycle ;
\draw  [fill={rgb, 255:red, 0; green, 0; blue, 0 }  ,fill opacity=1 ] (333.33,162.33) .. controls (333.33,160.12) and (335.12,158.33) .. (337.33,158.33) .. controls (339.54,158.33) and (341.33,160.12) .. (341.33,162.33) .. controls (341.33,164.54) and (339.54,166.33) .. (337.33,166.33) .. controls (335.12,166.33) and (333.33,164.54) .. (333.33,162.33) -- cycle ;
\draw  [fill={rgb, 255:red, 0; green, 0; blue, 0 }  ,fill opacity=1 ] (310.33,163.33) .. controls (310.33,161.12) and (312.12,159.33) .. (314.33,159.33) .. controls (316.54,159.33) and (318.33,161.12) .. (318.33,163.33) .. controls (318.33,165.54) and (316.54,167.33) .. (314.33,167.33) .. controls (312.12,167.33) and (310.33,165.54) .. (310.33,163.33) -- cycle ;
\draw [color={rgb, 255:red, 0; green, 0; blue, 0 }  ,draw opacity=1 ][line width=0.75]    (387.33,76.33) -- (464.33,132) ;
\draw [color={rgb, 255:red, 0; green, 0; blue, 0 }  ,draw opacity=1 ][line width=0.75]    (387.33,76.33) -- (390.33,132) ;
\draw [color={rgb, 255:red, 0; green, 0; blue, 0 }  ,draw opacity=1 ][line width=0.75]    (319.33,130) -- (387.33,76.33) ;
\draw  [fill={rgb, 255:red, 0; green, 0; blue, 0 }  ,fill opacity=1 ] (383.33,76.33) .. controls (383.33,74.12) and (385.12,72.33) .. (387.33,72.33) .. controls (389.54,72.33) and (391.33,74.12) .. (391.33,76.33) .. controls (391.33,78.54) and (389.54,80.33) .. (387.33,80.33) .. controls (385.12,80.33) and (383.33,78.54) .. (383.33,76.33) -- cycle ;
\draw  [fill={rgb, 255:red, 0; green, 0; blue, 0 }  ,fill opacity=1 ] (289.33,161.33) .. controls (289.33,159.12) and (291.12,157.33) .. (293.33,157.33) .. controls (295.54,157.33) and (297.33,159.12) .. (297.33,161.33) .. controls (297.33,163.54) and (295.54,165.33) .. (293.33,165.33) .. controls (291.12,165.33) and (289.33,163.54) .. (289.33,161.33) -- cycle ;
\draw    (547.33,163.67) -- (509,163.98) ;
\draw [shift={(507,164)}, rotate = 359.53] [color={rgb, 255:red, 0; green, 0; blue, 0 }  ][line width=0.75]    (10.93,-3.29) .. controls (6.95,-1.4) and (3.31,-0.3) .. (0,0) .. controls (3.31,0.3) and (6.95,1.4) .. (10.93,3.29)   ;
\draw    (528.33,130.67) -- (490,130.98) ;
\draw [shift={(488,131)}, rotate = 359.53] [color={rgb, 255:red, 0; green, 0; blue, 0 }  ][line width=0.75]    (10.93,-3.29) .. controls (6.95,-1.4) and (3.31,-0.3) .. (0,0) .. controls (3.31,0.3) and (6.95,1.4) .. (10.93,3.29)   ;
\draw [color={rgb, 255:red, 0; green, 0; blue, 0 }  ,draw opacity=1 ][line width=0.75]    (65.33,153.33) -- (70.33,120) ;
\draw [color={rgb, 255:red, 0; green, 0; blue, 0 }  ,draw opacity=1 ][line width=0.75]    (70.33,120) -- (88.33,152.33) ;
\draw [color={rgb, 255:red, 0; green, 0; blue, 0 }  ,draw opacity=1 ][line width=0.75]    (44.33,151.33) -- (70.33,120) ;
\draw  [fill={rgb, 255:red, 0; green, 0; blue, 0 }  ,fill opacity=1 ] (66.33,120) .. controls (66.33,117.79) and (68.12,116) .. (70.33,116) .. controls (72.54,116) and (74.33,117.79) .. (74.33,120) .. controls (74.33,122.21) and (72.54,124) .. (70.33,124) .. controls (68.12,124) and (66.33,122.21) .. (66.33,120) -- cycle ;
\draw  [fill={rgb, 255:red, 0; green, 0; blue, 0 }  ,fill opacity=1 ] (84.33,152.33) .. controls (84.33,150.12) and (86.12,148.33) .. (88.33,148.33) .. controls (90.54,148.33) and (92.33,150.12) .. (92.33,152.33) .. controls (92.33,154.54) and (90.54,156.33) .. (88.33,156.33) .. controls (86.12,156.33) and (84.33,154.54) .. (84.33,152.33) -- cycle ;
\draw  [fill={rgb, 255:red, 0; green, 0; blue, 0 }  ,fill opacity=1 ] (61.33,153.33) .. controls (61.33,151.12) and (63.12,149.33) .. (65.33,149.33) .. controls (67.54,149.33) and (69.33,151.12) .. (69.33,153.33) .. controls (69.33,155.54) and (67.54,157.33) .. (65.33,157.33) .. controls (63.12,157.33) and (61.33,155.54) .. (61.33,153.33) -- cycle ;
\draw  [fill={rgb, 255:red, 0; green, 0; blue, 0 }  ,fill opacity=1 ] (40.33,151.33) .. controls (40.33,149.12) and (42.12,147.33) .. (44.33,147.33) .. controls (46.54,147.33) and (48.33,149.12) .. (48.33,151.33) .. controls (48.33,153.54) and (46.54,155.33) .. (44.33,155.33) .. controls (42.12,155.33) and (40.33,153.54) .. (40.33,151.33) -- cycle ;
\draw    (149.33,150.67) -- (111,150.98) ;
\draw [shift={(109,151)}, rotate = 359.53] [color={rgb, 255:red, 0; green, 0; blue, 0 }  ][line width=0.75]    (10.93,-3.29) .. controls (6.95,-1.4) and (3.31,-0.3) .. (0,0) .. controls (3.31,0.3) and (6.95,1.4) .. (10.93,3.29)   ;
\draw   (320.33,117.67) .. controls (333.33,118.67) and (349.33,154.67) .. (347.33,166.67) .. controls (345.33,178.67) and (306.33,185.67) .. (285.33,168.67) .. controls (264.33,151.67) and (307.33,116.67) .. (320.33,117.67) -- cycle ;
\draw   (390.33,121.67) .. controls (403.33,122.67) and (421.33,154.67) .. (418.33,168.67) .. controls (415.33,182.67) and (368.33,188.67) .. (357.33,172.67) .. controls (346.33,156.67) and (377.33,120.67) .. (390.33,121.67) -- cycle ;
\draw   (465.33,121.67) .. controls (478.33,122.67) and (496.33,154.67) .. (493.33,168.67) .. controls (490.33,182.67) and (442.33,185.67) .. (431.33,169.67) .. controls (420.33,153.67) and (452.33,120.67) .. (465.33,121.67) -- cycle ;

\draw (560,155) node [anchor=north west][inner sep=0.75pt]   [align=left] {copies of $\displaystyle G_{0}$};
\draw (538,120) node [anchor=north west][inner sep=0.75pt]   [align=left] {copies of $\displaystyle G_{1}$};
\draw (527,65.4) node [anchor=north west][inner sep=0.75pt]    {$G_{2}$};
\draw (398,61.4) node [anchor=north west][inner sep=0.75pt]    {$r_{2}$};
\draw (290,97.4) node [anchor=north west][inner sep=0.75pt]    {$r_{1}^{( 1)}$};
\draw (364,101.4) node [anchor=north west][inner sep=0.75pt]    {$r_{1}^{( 2)}$};
\draw (455,97.4) node [anchor=north west][inner sep=0.75pt]    {$r_{1}^{( 3)}$};
\draw (159,141) node [anchor=north west][inner sep=0.75pt]   [align=left] {copies of $\displaystyle G_{0}$};
\draw (121,90) node [anchor=north west][inner sep=0.75pt]   [align=left] {$\displaystyle G_{1}$};
\draw (67,95.4) node [anchor=north west][inner sep=0.75pt]    {$r_{1}$};
\draw (20,159.4) node [anchor=north west][inner sep=0.75pt]    {$r_{0}^{( 1)}$};
\draw (54,163.4) node [anchor=north west][inner sep=0.75pt]    {$r_{0}^{( 2)}$};
\draw (89.33,160.73) node [anchor=north west][inner sep=0.75pt]    {$r_{0}^{( 3)}$};

\end{tikzpicture}

\end{figure}

%% file: appendix-pseudocodes.tex
\section{Pseudocode of the Implementation of \Cref{thm:main-decomp}}\label{appendix:pseudocode}

\begin{algorithm}[H]
  \DontPrintSemicolon
  Let $V(T) \gets U$, $E(T) \gets \emptyset$, $V_a \gets V \setminus V(T)$, and $G_a \gets G[V_a]$.
  
  Initialize $E_b$ by scanning all the neighbors of each $x \in U$.
  
  \While{$E_b$ is not empty}{
    Pop an arbitrary $(u,x) \in E_b$.
    \If{$u$ and $x$ are low-degree vertices in $T$ and they belong to different connected components of $T$}{
        $E(T) \gets E(T) \cup \{(u,x)\}$.
        
        Update $E_b$ and the disjoint-set union structure.
    }\ElseIf{$u$ is a low-degree vertex in $T$ and $x \in V_a$}{
        Let $T_u$ be the connected component of $T$ containing $u$.

        Let $W \gets \emptyset$.

        Let $\mathit{found} \gets \mathrm{DFS}(x,T_u,W)$ (\Cref{alg:DFS}).

        Regardless of the value of $\mathit{found}$, remove every vertex in $W$ from $V_a$, and update $G_a$, $E_b$, and the disjoint-set union structure to reflect the outcome of the DFS.
    }
    
  }
  Let $X = \{v \in V(T) \mid \deg_{T}(v) = 4\}$.
  
  \Return $(X, T)$.
  
  \caption{Our algorithm for \Cref{thm:main-decomp}.}\label{alg:indirect}
\end{algorithm}

\begin{algorithm}[H]
  \DontPrintSemicolon
    Mark $x$ as visited and add it to $W$.
    
    Scan the neighbors of $x$ in $G$ to determine whether some neighbor $v$ is a low-degree vertex of $T$ in a component different from $T_u$.

    \If{such $v$ exists}{
        Let $(x_1,x_2,\ldots,x_q)$ be the path from the root $x_1$ of the DFS tree to the current $x = x_q$.
    
        Add $P = (u,x_1,\ldots,x_q, v)$ to $T$ as follows:
        
        $V(T) \gets V(T) \cup \{x_1,x_2,\ldots,x_q\}$
        
        $E(T) \gets E(T) \cup \{(u,x_1), (x_1,x_2),\ldots, (x_q,v)\}$.

        \Return \textsc{true}.
    }

    \For{each neighbor $z$ of $x$ in $G_a$}{
        \If{$z$ is not marked visited}{
            \If{$\mathrm{DFS}(z,T_u,W)$ returns \textsc{true}}{
                \Return \textsc{true}.
            }
        }
    }

    \Return \textsc{false}.
    
  \caption{$\mathrm{DFS}(x,T_u,W)$, which returns whether it finds a linking path; $W$ is the set of vertices visited by the current DFS.}\label{alg:DFS}
\end{algorithm}